\documentclass[11pt]{article}

\usepackage[margin=1.15in]{geometry}
\usepackage{amsmath,amssymb,amsthm,mathtools}
\usepackage{aliascnt}
\usepackage{algorithm}
\usepackage[noend]{algpseudocode}
\usepackage[hidelinks]{hyperref}
\usepackage[nameinlink,capitalise]{cleveref}
\usepackage{xcolor}
\usepackage{microtype}

\newtheorem{theorem}{Theorem}
\newaliascnt{lemma}{theorem}
\newtheorem{lemma}[lemma]{Lemma}
\aliascntresetthe{lemma}
\newaliascnt{proposition}{theorem}
\newtheorem{proposition}[proposition]{Proposition}
\aliascntresetthe{proposition}
\newaliascnt{corollary}{theorem}
\newtheorem{corollary}[corollary]{Corollary}
\aliascntresetthe{corollary}
\newaliascnt{remark}{theorem}

\aliascntresetthe{remark}
\newaliascnt{definition}{theorem}

\aliascntresetthe{definition}

\newcommand{\OPT}{\operatorname{OPT}}
\newcommand{\eps}{\varepsilon}

\hypersetup{
  pdftitle={Approximating Combinatorial Contracts with Arbitrary Costs},
  pdfauthor={Xiaotie Deng, Hanyu Li, Chenghua Liu}
}

\title{\textbf{Approximating Combinatorial Contracts with Arbitrary Costs}}
\author{Xiaotie Deng\textsuperscript{1,2}\thanks{\href{mailto:csdeng@cityu.edu.hk}{\texttt{csdeng@cityu.edu.hk}}}\quad
Hanyu Li\textsuperscript{2}\thanks{\href{mailto:lhydave@pku.edu.cn}{\texttt{lhydave@pku.edu.cn}}}\quad
Chenghua Liu\textsuperscript{3}\thanks{\href{mailto:liuch.russell@gmail.com}{\texttt{liuch.russell@gmail.com}}}\\[0.5em]
\textsuperscript{1}Department of Computer Science, City University of Hong Kong\\[-0.1em]
\textsuperscript{2}CFCS, School of Computer Science, Peking University\\[-0.1em]
\textsuperscript{3}Institute of Software, Chinese Academy of Sciences}
\date{}

\begin{document}
\maketitle

\begin{abstract}
We study single-agent combinatorial contracts under linear payments. Under a
reward share $\alpha\in[0,1]$, an agent chooses a subset $S$ of $n$ hidden
actions, generating reward $f(S)$ at cost $c(S)$, to maximize
$\alpha f(S)-c(S)$, while the principal receives $(1-\alpha)f(S)$. For
nonnegative additive rewards and monotone supermodular
costs, D\"utting et al.\ (SODA 2026) proved an exponential supply-query lower
bound for exact optimization and left open whether a representation-independent
approximation is possible. We resolve this question in a stronger form,
removing all structural assumptions on the cost function: for nonnegative
additive rewards, arbitrary normalized nonnegative set costs, and every
$\varepsilon\in(0,1)$, we give a deterministic $(1-\varepsilon)$-approximation using
$O(n\log(n+1)/\varepsilon)$ supply queries, with query complexity independent
of numerical bit lengths and breakpoint separation. More generally, the same
guarantee and asymptotic query complexity hold for normalized monotone
subadditive rewards and arbitrary normalized nonnegative costs under exact
best-response and value-query access. The key idea is a cost-independent,
reward-side scale certificate. The algorithm tries all singleton rewards as
candidate anchors; one of them brackets the optimal retained share within a
factor $n^2$. Geometric search and monotonicity of the induced response reward
then yield the approximation without enumerating best-response breakpoints.
\end{abstract}

\section{Introduction}

Contract design governs how one party incentivizes another to act on its
behalf.  A central difficulty is \emph{hidden action}: the principal observes
the outcome of a delegated task, but not the costly actions taken by the
agent to produce it.  The principal must therefore base payment on the
observable outcome while anticipating how the payment changes the agent's
behavior.  This classical economic problem has recently become a subject of
growing interest in theoretical computer science and algorithmic game
theory~\cite{DFT24}.

Many delegated tasks involve several actions that may be taken together,
rather than a single choice between effort and no effort.  The
combinatorial-action model introduced by D\"utting, Ezra, Feldman, and
Kesselheim~\cite{DEFK21} captures this setting.  A project has $n$ available
actions, indexed by $[n]=\{1,\ldots,n\}$, and the agent may take any subset
$S\subseteq[n]$.  The set function $f(S)$ is the principal's expected gross
reward when $S$ is taken, while $c(S)$ is the agent's cost.  In the standard
binary-outcome interpretation, the project either succeeds or fails and,
after normalizing the reward from success to one, $f(S)$ is its success
probability.  The set functions allow actions to be substitutes or
complements.  They also create a computational challenge: there are $2^n$
possible action sets.

Because the principal cannot observe $S$, payment cannot depend directly on
the chosen actions.  A linear contract is specified by a share
$\alpha\in[0,1]$: the agent receives an $\alpha$ fraction of the generated
reward.  Facing this contract, the agent chooses a \emph{best response}
\[
 S_\alpha\in\arg\max_{S\subseteq[n]}
                  \{\alpha f(S)-c(S)\}.
\]
Thus the response $S_\alpha$ is simply the action set that maximizes the
agent's expected payment minus cost.  Ties are resolved in the principal's
favor, as formalized in \cref{sec:model}.  Anticipating the response, the
principal receives
\[
             P(\alpha)=(1-\alpha)f(S_\alpha)
\]
and seeks a contract maximizing this quantity.  We denote the optimal
principal utility by $\OPT=P^*=\max_{\alpha\in[0,1]}P(\alpha)$.

An explicit description of a general set function may itself contain
exponentially many numbers.  We therefore work with oracle access.  A value
query evaluates $f(S)$ or $c(S)$, while a best-response query returns
$S_\alpha$ for a specified contract.  For additive rewards, let $v_i$ be the
reward of action $i$, so $f(S)=\sum_{i\in S}v_i$.  Best responses can then be implemented by a standard
supply query to the cost function~\cite{DFGR26}: given an item-price vector
$p$, it returns a set maximizing
\[
                   \sum_{i\in S}p_i-c(S).
\]
Setting $p_i=\alpha v_i$ yields $S_\alpha$.  Supply access thus lets the
algorithm observe the economically relevant response without requiring an
explicit representation of the cost function.

As $\alpha$ varies, the agent's response changes only at
\emph{breakpoints}.  General procedures can enumerate them with oracle
access in time polynomial in their number~\cite{DDPP24,DFG24}.  That number
can nevertheless be exponential~\cite{DEFK21}, and the relevant
breakpoints can be arbitrarily close~\cite{DFGR26}.  Exact optimization asks
for the best breakpoint.  Approximation poses a different question: can the
principal obtain nearly optimal utility without identifying the exact
response change that attains the optimum?  We call an FPTAS
\emph{representation independent} if, for every $\eps\in(0,1)$, it makes a
number of oracle calls polynomial in $n$ and $1/\eps$, independent of
numerical bit lengths, scales, and breakpoint separation.  This is an
exact-real query model: oracle inputs and answers are exact real numbers,
and complexity counts oracle calls rather than bit operations.

D\"utting, Feldman, Gal-Tzur, and Rubinstein~\cite{DFGR26} studied the query
and communication complexity of this problem across several combinations
of submodular and supermodular rewards and costs.  Most relevant here, they
consider additive rewards and monotone supermodular costs.  Additivity means
that every action has an independent reward $v_i$, whereas supermodularity
means that marginal costs increase as more actions are taken.  They prove
that any algorithm computing the exact optimal contract in this setting
requires $2^{\Omega(n)}$ supply queries.

Their tightness discussion leaves a sharp approximation question.  A
weakly polynomial FPTAS, whose complexity depends on the numerical bit
lengths, was known~\cite{DEFK21}.  Removing this dependence, the
representation-independent FPTAS for additive costs in~\cite{DEFK25}
extends to monotone subadditive costs~\cite[Appendix~B]{DFGR26}.  The latter work
explicitly left open
whether instances with supermodular costs, even with additive rewards,
admit any representation-independent approximation, even under
best-response access.

The known approximation method~\cite{DEFK25,DFGR26} has an appealingly
simple form: compute welfare, guess one action from the optimal response,
and search a geometric grid of contracts.  Its scale certificate, however,
is a singleton
\emph{cost}, and the argument relating that cost to the cost of the entire
response relies on additive or subadditive costs.  Supermodular costs reverse
the relevant inequality, while arbitrary costs provide no such comparison.
This is the precise obstruction behind the open problem.

We show that the exponential barrier is specific to exact optimization.
For the class in the lower bound of~\cite{DFGR26}, exact
optimization requires $2^{\Omega(n)}$ supply queries, whereas a
$(1-\eps)$-approximation requires only
$O(n\log(n+1)/\eps)$ such queries.  In fact, the approximation does not need
supermodularity, monotonicity, or any other structure on the cost function.

\begin{theorem}[Supply-oracle resolution]\label{thm:supply}
Let $f(S)=\sum_{i\in S}v_i$ be additive with $v_i\geq0$, and let
$c(\varnothing)=0$ and $c(S)\geq0$ for every $S\subseteq[n]$.  For every
$\eps\in(0,1)$, a deterministic algorithm returns a contract with principal
utility at least $(1-\eps)\OPT$ using
\[
          O\!\left(\frac{n\log(n+1)}{\eps}\right)
\]
standard supply queries to $c$, $O(n)$ reward-value queries, and one
cost-value query.  The bound is independent of numerical values and
breakpoint separation.
\end{theorem}

This theorem directly resolves the open problem in~\cite{DFGR26}.  Its
proof reveals that the relevant structure lies on
the reward side.  A set function is normalized if it is zero on the empty
set.  Recall that $f$ is subadditive if
$f(S\cup T)\leq f(S)+f(T)$ for every $S,T\subseteq[n]$.
The supply-oracle result is a special case of the following statement.

\begin{theorem}[General reward theorem]\label{thm:general}
Let $f:2^{[n]}\to\mathbb{R}_{\geq0}$ be normalized, monotone, and
subadditive, and let $c(\varnothing)=0$ and $c(S)\geq0$ for every
$S\subseteq[n]$.  For every $\eps\in(0,1)$, there is a deterministic
algorithm that returns a contract whose principal utility is at least
$(1-\eps)\OPT$.  It uses
\[
          O\!\left(\frac{n\log(n+1)}{\eps}\right)
\]
exact best-response and value queries.
More precisely, these are reward-value queries; one additional cost-value
query is used to evaluate welfare.
The bound is independent of numerical values and breakpoint separations.
\end{theorem}

Additivity is needed only to implement a best-response query as a standard
supply query to the cost function.  Combining \cref{thm:supply} with the
exact lower bound gives the following separation.

\begin{corollary}[Exact--approximate separation]\label{cor:separation}
In the exact-real supply-oracle model, combinatorial contracts with additive
rewards and monotone supermodular costs exhibit an exponential separation
between exact and approximate query complexity: exact optimization requires
$2^{\Omega(n)}$ supply queries in the worst case, whereas a
$(1-\eps)$-approximation uses
$O(n\log(n+1)/\eps)$ supply queries.
\end{corollary}

Our proof retains the simple welfare--guess--search template, but changes
what fixes the scale.  Instead of using a singleton cost, we anchor the
search by a singleton \emph{reward}.  This makes the argument independent
of the cost structure.

In the nontrivial case $\OPT>0$, let $\alpha^*$ be an optimal contract, let
$S^*$ be its induced response, and let $\delta^*=1-\alpha^*$ be the
principal's optimal retained share.  Our analysis
combines two structural observations.
First, a welfare--profit comparison relates the maximum welfare, or total
surplus, to the principal's optimal utility:
\[
       \OPT\leq W\leq n\OPT,
       \qquad W=\max_{S\subseteq[n]}\{f(S)-c(S)\}.
\]
Second, if $j^*$ is a largest-singleton-reward action in $S^*$, monotonicity
and subadditivity give
\[
       f(\{j^*\})\leq f(S^*)\leq nf(\{j^*\}).
\]
Since $\OPT=\delta^*f(S^*)$, the two observations yield
\[
 \frac{W}{n^2f(\{j^*\})}
       \leq \delta^*
       \leq \min\!\left\{1,\frac{W}{f(\{j^*\})}\right\}.   \tag{1}
\]

The algorithm tries each of the $n$ singleton rewards as the unknown
$f(\{j^*\})$.  One guess therefore produces a multiplicative interval of
width at most $n^2$ containing $\delta^*$.  A geometric grid covers that
interval with $O(\log(n+1)/\eps)$ contracts.  Monotonicity of the reward
induced by the agent's response then shows that one grid point gives a
$(1-\eps)$-approximation.  In contrast to exact optimization, the algorithm
never enumerates response breakpoints, searches over the numerical range of
costs or rewards, or separates cases according to the cost class.  The
lower bound hides the identity of the best breakpoint; the approximation
needs only its multiplicative scale.

As follow-up to this open-problem resolution, we clarify why approximation
escapes the exact lower bound and where this advantage ends.  On the
equal-revenue family, ordinary approximation needs only the
multiplicative scale of the optimal breakpoint, whereas relative error
$\gamma_n=\Theta(2^{-n})$ forces identification of its hidden identity and
restores exponential query complexity.  Separately, we show robustness to
additive best-response error $\tau$, with guarantee
$(1-\eps)\OPT-3\tau/\eps$, and establish tightness of the scale certificate.

This paper is organized as follows.  \Cref{sec:related} reviews related work,
\cref{sec:model} fixes the oracle model, and \cref{sec:algorithm} develops the
reward-anchored scale certificate and the resulting FPTAS\@.
\Cref{sec:followup} presents the follow-up results on robustness,
high-accuracy query complexity, and tightness; their proofs appear in the
appendices.  \Cref{sec:conclusion} concludes with a discussion of the
remaining boundary.

\section{Related work}\label{sec:related}

The original work on the combinatorial-action model~\cite{DEFK21} gave an
exact algorithm for gross-substitutes rewards, proved NP-hardness for
submodular rewards, and gave a weakly polynomial FPTAS for general monotone
rewards and costs with best-response access.  Ezra, Feldman, and
Schlesinger~\cite{EFS24} strengthened the computational hardness: with
submodular rewards and additive costs, no polynomial-time value-query
algorithm obtains a constant-factor approximation unless $\mathsf{P} =
\mathsf{NP}$.  Complementing these hardness results, breakpoint-enumeration
algorithms run in time polynomial in the number of breakpoints and give
exact tractability for supermodular rewards and submodular
costs~\cite{DDPP24,DFG24}.

D\"utting, Ezra, Feldman, and Kesselheim~\cite{DEFK25} subsequently gave a
representation-independent FPTAS for general monotone rewards and additive
costs; its welfare--profit comparison, welfare query, item guess, and
geometric search motivate part of our analysis.  D\"utting, Feldman,
Gal-Tzur, and Rubinstein~\cite{DFGR26} observed that this approach extends to
monotone subadditive costs and established the exact supply-query lower
bound for supermodular costs.  Our result retains the
representation-independent query bound while removing every structural
assumption on the cost function.

Recent work studies exact optimization for special reward and cost
classes~\cite{FY25}, demand-type representations~\cite{BaldwinEtAl26}, and
inapproximability for budgeted variants~\cite{FeldmanEtAl26}.  None of these
results gives the subadditive-reward, arbitrary-cost oracle guarantee proved
here.

\section{Model, objective, and oracle access}\label{sec:model}

There are $n$ hidden actions, and the agent may take any set $S\subseteq[n]$.
Taking $S$ creates expected gross reward $f(S)$ for the principal and cost
$c(S)$ for the agent.  The principal cannot pay for the unobserved actions
directly; it can only specify how the observable reward will be divided.

\paragraph{Contract, response, and objective.}
The principal first chooses a linear contract $\alpha\in[0,1]$, promising the
agent an $\alpha$ share of the generated reward.  The agent's utility from
$S$ is therefore $\alpha f(S)-c(S)$, so the contract induces the best response
\[
 S_\alpha\in\arg\max_{S\subseteq[n]}
        \{\alpha f(S)-c(S)\}.
\]
Ties are broken toward larger $f(S)$, as preferred by the principal.  The
principal keeps the remaining share $1-\alpha$; its utility and optimum are
\[
 P(\alpha)=(1-\alpha)f(S_\alpha),\qquad
          P^*=\OPT=\max_{\alpha\in[0,1]}P(\alpha).
\]
Given an accuracy parameter $\eps\in(0,1)$, our algorithm must return a
contract $\widehat\alpha$ with
$P(\widehat\alpha)\geq(1-\eps)P^*$.  The output is the contract parameter;
the agent subsequently selects $S_{\widehat\alpha}$.

For our main result, the reward $f:2^{[n]}\to\mathbb{R}_{\geq0}$ is normalized, monotone, and
subadditive: $f(\varnothing)=0$, adding actions cannot decrease reward, and
$f(S\cup T)\leq f(S)+f(T)$.  The cost
$c:2^{[n]}\to\mathbb{R}_{\geq0}$ is arbitrary except that
$c(\varnothing)=0$.  The maximum total surplus is
\[
              W=\max_{S\subseteq[n]}\{f(S)-c(S)\}
\]
and serves only as a scale benchmark; the principal's objective remains
$P(\alpha)$.  Since the agent maximizes $f(S)-c(S)$ at $\alpha=1$, one
best-response query computes $W$.

Because $f$ and $c$ may have exponential-size descriptions, the algorithm
uses exact queries: a best-response query at $\alpha$ returns $S_\alpha$,
while value queries return $f(S)$ or $c(S)$.  Complexity counts oracle calls,
not numerical bit lengths.

For additive rewards, write $f(S)=\sum_{i\in S}v_i$, where $v_i=f(\{i\})$.
A best-response query at $\alpha>0$ is the standard supply query to $c$ at
prices $p_i=\alpha v_i$:
\[
       S\in\arg\max_T\left\{\sum_{i\in T}p_i-c(T)\right\},
\]
The standard larger-cost tie-breaking agrees with principal-favoring
tie-breaking because tied sets satisfy
$c(S)-c(T)=\alpha(f(S)-f(T))$.

\section{Reward-anchored geometric search}\label{sec:algorithm}

The algorithm searches for the scale of the principal's optimal retained
share.  A contract whose retained share lies just below the optimal one is
approximately optimal because the reward of a best response grows with the
agent's share.  The main difficulty is to locate this scale without numerical
or structural assumptions on costs.  We first formalize the approximation
target, then derive a cost-independent welfare benchmark and a reward-side
scale certificate, and finally turn the certificate into a geometric search.

\subsection{From approximation to scale localization}

It is useful to parameterize a contract by
\[
                         \delta=1-\alpha,
\]
the fraction of generated reward retained by the principal.  Fix an optimal
contract--response pair $(\alpha^*,S^*)$ and write
$\delta^*=1-\alpha^*$, so that $P^*=\delta^*f(S^*)$.

The starting point of the known search framework is response-reward
monotonicity.

\begin{lemma}[Response-reward monotonicity~\cite{DEFK21,DEFK25,DFGR26}]
\label{lem:monotone}
If $0\leq\alpha\leq\alpha'\leq1$, then
$f(S_\alpha)\leq f(S_{\alpha'})$.
\end{lemma}

The response sets themselves need not be nested.  The intuition is that a
larger agent share places greater weight on generated reward, so the reward
of a utility-maximizing response cannot decrease.  Formally, the statement
follows immediately by comparing the two best-response inequalities; we
use it as a standard fact from the earlier framework.

This monotonicity identifies exactly what the search must find.  If a query
uses a retained share satisfying
\[
                 (1-\eps)\delta^*\leq\delta\leq\delta^*, \tag{2}
\]
then its agent share $\alpha=1-\delta$ is at least $\alpha^*$.  Hence
$f(S_\alpha)\geq f(S^*)$, and its principal utility satisfies
\[
 P(\alpha)=\delta f(S_\alpha)
 \geq(1-\eps)\delta^*f(S^*)
 =(1-\eps)P^*.
\]
Thus approximation reduces to locating the multiplicative scale of the
unknown $\delta^*$; the identity of the exact optimal breakpoint is not
needed.

A uniform additive grid is insufficient because $\delta^*$ has no
inverse-polynomial lower bound.  A geometric grid avoids dependence on
numerical magnitudes, provided it starts from a polynomial-width
multiplicative interval containing $\delta^*$.  Indeed, if
$b/R\leq\delta^*\leq b$ for $R=\operatorname{poly}(n)$, then the points
$b(1-\eps)^k$ contain one satisfying~(2) after
$O(\log R/\eps)$ steps.  This is the geometric-search step used in prior
work~\cite{DEFK25,DFGR26}; here we invoke it rather than reprove it as a
separate abstract lemma.

Earlier algorithms obtain the starting interval by guessing the cost of an
action in $S^*$.  Additive or subadditive costs then relate that singleton
cost to $c(S^*)$~\cite{DEFK25,DFGR26}.  With arbitrary set costs, singleton
costs provide no information about the cost of the whole response.  Our task
is therefore to construct an enumerable anchor for $\delta^*$ using rewards
instead.

\subsection{The reward-anchored scale certificate}

Recall that
\[
           W=\max_{S\subseteq[n]}\{f(S)-c(S)\}.
\]
Equivalently, $W$ is the agent's maximum utility at $\alpha=1$.
The lower comparison $P^*\leq W$ is immediate: the agent can always choose
the empty set, so an induced response has nonnegative agent utility, and its
welfare is at least the principal's utility.  The substantive direction is
the following upper bound.  Its proof views the agent's maximum utility as
an upper envelope of lines.  As the offered share increases, we charge
welfare growth when a new action first appears on that envelope.  At most
$n$ actions can appear for the first time, and each charge is bounded by the
principal's optimal utility.

\begin{lemma}\label{lem:welfare}
For normalized monotone subadditive rewards and arbitrary normalized
nonnegative costs,
\[
                         P^*\leq W\leq nP^*.
\]
\end{lemma}

\begin{proof}
Let $(\alpha^*,S^*)$ be optimal.  Since the empty set is available,
$c(S^*)\leq\alpha^*f(S^*)$, and hence
\[
 P^*=(1-\alpha^*)f(S^*)
 \leq f(S^*)-c(S^*)
 \leq W.
\]
If $W=0$, the claim is immediate.  Assume henceforth that $W>0$.
Let $g(\alpha)$ denote the agent's maximum utility under contract $\alpha$.
As a function of the offered share, it is the convex, piecewise-linear upper
envelope
\[
        g(\alpha)=\max_S\{\alpha f(S)-c(S)\}.
\]
Because $\varnothing$ is available and costs are nonnegative, $g(0)=0$;
also $g(1)=W$.  Let
\[
0=a_0<a_1<\cdots<a_k<a_{k+1}=1
\]
be the endpoints of its positive-length linear pieces, and let $S_t$ be an
active response on $[a_t,a_{t+1})$.  Choose the representatives at endpoints
consistently with principal-favoring tie-breaking.  The slopes
$f(S_t)$ are nondecreasing, and integration gives
\[
 W=g(1)-g(0)
   =\sum_{t=0}^k(a_{t+1}-a_t)f(S_t).                      \tag{3}
\]

We now expose actions only when they first occur on the response curve.
Let $\ell_1$ be the first index with $S_{\ell_1}\ne\varnothing$.  Having
chosen $\ell_r$, let the cumulative set of actions seen so far be
\[
 T_r=\bigcup_{t=0}^{\ell_r}S_t
\]
and let $\ell_{r+1}$ be the first index larger than $\ell_r$ for which
$S_{\ell_{r+1}}\nsubseteq T_r$.  If this process produces $d$ epochs, set
$\ell_{d+1}=k+1$ and $T_0=\varnothing$.  Every epoch introduces a new
action, so $d\leq n$.  Moreover, for
$\ell_r\leq t<\ell_{r+1}$ we have $S_t\subseteq T_r$.

Let $B_r=T_r\setminus T_{r-1}$ be the block of newly appearing actions.
Monotonicity and subadditivity imply
\[
f(S_t)\leq f(T_r)\leq\sum_{j=1}^r f(B_j)
\qquad(\ell_r\leq t<\ell_{r+1}).
\]
Thus $[a_{\ell_r},a_{\ell_{r+1}})$ is an epoch during which every response
uses only
$T_r=B_1\mathbin{\dot\cup}\cdots\mathbin{\dot\cup}B_r$.
Responses before the first epoch are empty.  Partitioning~(3) into these
epochs, applying the preceding inequalities, and exchanging sums gives
\begin{align*}
W
&\leq
\sum_{r=1}^d
(a_{\ell_{r+1}}-a_{\ell_r})
\sum_{j=1}^r f(B_j)\\
&=
\sum_{j=1}^d(1-a_{\ell_j})f(B_j).
\end{align*}
No intermediate response introduces a new action.  Hence every action in
$B_j$ first appears in $S_{\ell_j}$, so $B_j\subseteq S_{\ell_j}$.
Monotonicity therefore yields
\[
(1-a_{\ell_j})f(B_j)
\leq(1-a_{\ell_j})f(S_{\ell_j})
\leq P^*.
\]
Summing over $d\leq n$ blocks proves the claim.
\end{proof}

The upper-bound proof adapts the response-curve argument
of~\cite{DEFK25}.  The extension important here is that it uses no
additivity or other structure of the cost function.

The welfare benchmark controls $P^*$ without revealing the unknown response
$S^*$.  To recover the missing reward scale, choose an action of $S^*$ with
largest singleton reward.  Subadditivity makes this singleton an
$n$-approximation to $f(S^*)$, while all $n$ possible singleton rewards can
be queried in advance.  Combining these observations gives the certificate
that replaces the cost anchor used in previous work.

\begin{lemma}[Reward-anchored scale]\label{lem:bracket}
Suppose $P^*>0$.  Let $(\alpha^*,S^*)$ be optimal and let
$\delta^*=1-\alpha^*$ be the principal's optimal retained share.  Choose an
action $j^*\in S^*$ with largest singleton reward:
\[
         j^*\in\arg\max_{j\in S^*}f(\{j\}).
\]
Then
\[
\frac{W}{n^2f(\{j^*\})}
\leq\delta^*
\leq
\min\!\left\{1,\frac{W}{f(\{j^*\})}\right\}.              \tag{4}
\]
The ratio between the upper and lower endpoints is at most $n^2$.
\end{lemma}

\begin{proof}
Since $P^*>0$, the set $S^*$ contains an item of positive singleton reward:
otherwise subadditivity would give $f(S^*)=0$.  Monotonicity,
subadditivity, and the definition of $j^*$ give
\[
 f(\{j^*\})\leq f(S^*)\leq
 \sum_{j\in S^*}f(\{j\})\leq nf(\{j^*\}).                 \tag{5}
\]
By definition,
\[
                         P^*=\delta^*f(S^*).              \tag{6}
\]
The right inequality in~(4) follows from \cref{lem:welfare}, (5), and~(6):
$\delta^*f(\{j^*\})\leq P^*\leq W$, while $\delta^*\leq1$.

For the left inequality, \cref{lem:welfare} and~(5) imply
\[
        W\leq nP^*
          =n\delta^*f(S^*)
          \leq n^2\delta^*f(\{j^*\}).
\]
Finally, if $W/f(\{j^*\})\leq1$, the endpoint ratio is exactly $n^2$;
otherwise it is $n^2f(\{j^*\})/W<n^2$.
\end{proof}

\subsection{The resulting FPTAS}

The oracle model supplies exactly the observable quantities in the
certificate.  A best-response query at $\alpha=1$ returns a
welfare-maximizing action set
$T\in\arg\max_S\{f(S)-c(S)\}$, from which reward and cost values give
$W=f(T)-c(T)$.  Singleton reward queries enumerate all possible anchors.

The algorithm now directly implements the preceding scale bracket.
One query at $\alpha=1$ computes welfare, $n$ singleton queries list all
possible reward anchors, and one geometric grid is searched for each
positive anchor.  No response breakpoint is computed or approximated.

We parameterize a contract by the principal's retained share
$\delta=1-\alpha$.  Let $q=1-\eps$ be the geometric contraction ratio, and
let $K$ be the number of grid steps needed to span a factor $n^2$, namely the
least nonnegative integer such that
\[
                         q^K\leq\frac{1}{n^2}.             \tag{7}
\]
The shift by one grid position in \cref{alg:fptas} ensures that every queried
contract has $\alpha>0$.  This is immaterial for the general best-response
oracle and ensures correct implementation by the standard supply oracle in
the additive-reward special case.

\begin{algorithm}[H]
\caption{Reward-anchored geometric search}\label{alg:fptas}
\begin{algorithmic}[1]
\Require Actions $[n]$; value oracles for a normalized monotone subadditive
         reward $f$ and a normalized nonnegative cost $c$.
\Require An exact best-response oracle returning $S_\alpha$ at each queried
         contract $\alpha\in[0,1]$; accuracy $\eps\in(0,1)$.
\Ensure A contract $\widehat\alpha\in[0,1]$.
\State Query each action's singleton reward $a_i=f(\{i\})$, $i\in[n]$.
\State Query the best-response oracle at $\alpha=1$, denote its response by
       $T$, and compute $W=f(T)-c(T)$.
\If{$W=0$}
   \State \Return $\alpha=1$.
\EndIf
\State Let $q=1-\eps$ and choose the least $K\geq0$ satisfying~(7).
\For{every $i$ with $a_i>0$}
   \State Set the candidate bracket's upper endpoint
          $b_i\gets\min\{1,W/a_i\}$.
   \For{$k=1,\ldots,K+1$}
      \State Set the retained share $\delta_{i,k}\gets b_i q^k$ and
             contract share $\alpha_{i,k}\gets1-\delta_{i,k}$.
      \State Query $S_{\alpha_{i,k}}$ and its value
             $f(S_{\alpha_{i,k}})$.
      \State Record $\delta_{i,k}f(S_{\alpha_{i,k}})$.
   \EndFor
\EndFor
\State \Return a queried contract with maximum recorded principal utility.
\end{algorithmic}
\end{algorithm}

\begin{proof}[Proof of \cref{thm:general,thm:supply}]
If $W=0$, every best response satisfies
\[
P(\alpha)
\leq f(S_\alpha)-c(S_\alpha)
\leq W=0,
\]
where the first inequality follows from
$c(S_\alpha)\leq\alpha f(S_\alpha)$ because $\varnothing$ is available.
Thus the returned contract is optimal.

Assume $W>0$ and consider the iteration $i=j^*$ supplied by
\cref{lem:bracket}.  Let $b=b_{j^*}$ be the correct bracket's upper endpoint.
The bracket gives
$b\geq\delta^*$ and
\[
\frac{b}{n^2}
\leq\frac{W}{n^2f(\{j^*\})}
\leq\delta^*.                                               \tag{8}
\]
By~(7), the last grid point satisfies
$bq^{K+1}\leq q\delta^*<\delta^*$.

If $b=\delta^*$, the first queried point is $q\delta^*$.
Otherwise, let $k$ be the first queried grid index with
$bq^k\leq\delta^*$.  The preceding point is larger than $\delta^*$, so
\[
              q\delta^*<bq^k\leq\delta^*.                  \tag{9}
\]
In either case the algorithm queries some $\delta$ satisfying
$q\delta^*\leq\delta\leq\delta^*$.
For its contract $\alpha=1-\delta$, we have $\alpha\geq\alpha^*$.
\Cref{lem:monotone} therefore implies
$f(S_\alpha)\geq f(S^*)$, and hence
\[
P(\alpha)
 =\delta f(S_\alpha)
 \geq q\delta^*f(S^*)
 =(1-\eps)P^*.
\]
The algorithm returns the best queried contract.

It remains to count queries.  Since
$-\log(1-\eps)\geq\eps$,
\[
K
\leq\left\lceil\frac{2\log n}{\eps}\right\rceil
\quad (n\geq2).
\]
Thus there are $O(n\log(n+1)/\eps)$ best-response and value queries.
If $f$ is additive, all queried response values are computed from the
$n$ singleton values.  Every queried $\delta$ is strictly smaller than one,
so the corresponding $\alpha>0$ best response is implemented by the supply
query with prices $p_j=\alpha f(\{j\})$ and the correct tie-breaking.  Only
one additional cost-value query, used to compute $W$, is required.
\end{proof}

\section{Beyond the FPTAS: robustness and limits}\label{sec:followup}

The efficiency of the FPTAS rests on two features of the model.  It uses
exact best responses, and it searches only for the multiplicative scale of
the optimal retained share rather than the identity of its breakpoint.  The
first feature is robust: additive response error produces only an additive
loss in principal utility.  The second has genuine limits: exponentially
fine accuracy requires recovering the hidden breakpoint identity, and the
$n^2$ localization width cannot be improved in general.

\subsection{Robustness to approximate best responses}\label{sec:robust}

The reward-anchored search continues to work when the agent's response is
computed only approximately.  A
$\tau$-best-response oracle, queried at $\alpha$, may return any set
$\widetilde S_\alpha$ satisfying
\[
\alpha f(\widetilde S_\alpha)-c(\widetilde S_\alpha)
\geq
\max_S\{\alpha f(S)-c(S)\}-\tau.
\]
The guarantee below is for the response selection implemented by the oracle;
no consistency across queries is required.  We compare the realized principal
utility $(1-\alpha)f(\widetilde S_\alpha)$ with the exact-response benchmark
$P^*$ defined in \cref{sec:model}.

\begin{theorem}[Approximate-oracle robustness]\label{thm:robust}
Under the assumptions of \cref{thm:general}, suppose best-response queries
are replaced by $\tau$-best-response queries and exact reward and cost
values remain available.  For every $\eps\in(0,1)$, a deterministic
algorithm makes
\[
O\!\left(\frac{n\log(n+1)}{\eps}\right)
\]
$\tau$-best-response and reward-value queries, plus one cost-value query, and
returns a queried contract and response with principal utility at least
\[
                         (1-\eps)P^*-\frac{3\tau}{\eps}.
\]
\end{theorem}

To see the proof idea, query the approximate oracle at $\alpha=1$.  Its
realized welfare gives an additive-$\tau$ estimate of $W$, which widens the
scale bracket by only a constant factor.  The grid then deliberately moves
one step below the target retained share.  The resulting separation from
$\alpha^*$ creates enough agent-utility margin to absorb the response error,
at a principal-utility cost of order $\tau/\eps$.  Appendix~\ref{app:robust}
gives the full argument and constants.

\subsection{Limits of scale localization}\label{sec:limits}

The equal-revenue family used for the exact lower bound
of~\cite{DFGR26} has exponentially many critical contracts with the same
principal utility.  Its hidden perturbation changes the identity of the
unique optimum while barely changing its scale.  Ordinary approximation can
ignore this identity, but accuracy comparable to the perturbation cannot.

\begin{theorem}[Approximation requires accuracy dependence]
\label{thm:approx-lower}
There is a sequence $\gamma_n=\Theta(2^{-n})$ such that any randomized
algorithm which, with constant success probability, computes a
$(1-\gamma_n)$-approximate contract for every instance with additive reward
and monotone supermodular cost requires $2^{\Omega(n)}$ supply queries in
expectation.  The conclusion continues to hold with polynomially many
cost-value queries.
\end{theorem}

This establishes a high-accuracy boundary, not a lower bound for constant or
inverse-polynomial accuracy.  The construction and proof appear in
Appendix~\ref{app:high-accuracy}.

A separate boundary concerns the width of the certificate itself.  Its
$n^2$ factor combines $W\leq nP^*$ with
$f(S^*)\leq n\max_{i\in S^*}f(\{i\})$.  The following proposition shows that
the two losses can occur simultaneously, while the first remains tight even
in the monotone-supermodular cost class.

\begin{proposition}[Tightness of reward-anchored localization]
\label{prop:tightness}
\begin{enumerate}
\item[(i)] For every $m\geq2$ and fixed $\eta\in(0,1)$, there is an instance
with $n=2m$ actions, additive rewards, and arbitrary normalized nonnegative
costs having a unique optimal principal contract $(\alpha^*,S^*)$ such that
\[
\frac{W}{(1-\alpha^*)\max_{i\in S^*}f(\{i\})}
=\frac{m(1+\eta+m/2)}{1+\eta}=\Omega(n^2).
\]
\item[(ii)] For every $n\geq2$ and $r>1$, there is an instance with
normalized additive rewards and monotone supermodular costs satisfying
\[
                         \frac{W}{P^*}=n-\frac{n-1}{r}.
\]
\end{enumerate}
\end{proposition}

The constructions are given in Appendix~\ref{app:tightness}.  They establish
tightness of the localization certificate, not a matching query lower bound,
because the bracket width enters the algorithm only logarithmically.

\section{Conclusion and discussion}\label{sec:conclusion}

We resolve the representation-independent approximation question for
additive rewards and arbitrary set costs, and more generally for monotone
subadditive rewards under best-response access.  The resulting separation
between exact and approximate query complexity is explained by reward-side
localization: approximation needs the multiplicative scale of an optimal
breakpoint, while exact optimization may require its hidden identity.  The
robustness and tightness results show both the reach and the limits of this
explanation.

The main remaining question is whether representation-independent
approximation with arbitrary costs extends to monotone rewards beyond
subadditivity.  Our proof uses subadditivity twice: to obtain the polynomial
welfare--profit comparison $W\leq nP^*$ and to approximate the reward of an
optimal response by one of its singleton rewards.  Progress beyond this
boundary therefore requires a different polynomial welfare gap, a different
enumerable anchor, or an impossibility result showing that no such
localization is possible.

\section*{Acknowledgments}
OpenAI Codex assisted with proof development and verification.  The authors
independently checked all arguments and assume full responsibility for the
content of this manuscript.

\clearpage
\appendix
\small

\section{Proof of approximate-response robustness}\label{app:robust}

\begin{proof}[Proof of \cref{thm:robust}]
Query the approximate oracle at $\alpha=1$, let $T$ be its response, and let
$L=\max\{0,f(T)-c(T)\}$ be the resulting lower estimate of optimal welfare.
Since $W$ is optimal welfare,
\[
                         W-\tau\leq f(T)-c(T)\leq W.
\]
If $L\leq\tau$, then $W\leq L+\tau\leq2\tau$.  Returning $\alpha=1$
gives zero principal utility, which already satisfies the claimed bound
because $P^*\leq W$.

Assume $L>\tau$ and let $U=L+\tau$ be the corresponding welfare upper bound.
Then
\[
                         L\leq W\leq U<2L.
\]
Let $(\alpha^*,S^*)$ be an exact optimal contract--response pair, let
$\delta^*=1-\alpha^*$ be its retained share, and let $a=f(\{j^*\})$ be the largest singleton
reward in $S^*$.  As in \cref{lem:bracket},
\[
                    \frac{W}{n^2a}\leq\delta^*
                    \leq\min\left\{1,\frac{W}{a}\right\}.
\]
Let $b=\min\{1,U/a\}$ be the widened bracket's upper endpoint.  The preceding
inequalities imply
\[
                         \frac{b}{2n^2}\leq\delta^*\leq b. \tag{10}
\]
Indeed, the upper bound follows from $W\leq U$.  For the lower bound,
$U<2W$; this proves the claim directly when $b=U/a$, and when $b=1$ it
gives $W>a/2$.

Allocate $\rho=\eps/3$ of the error budget to oracle error and let
$q=1-\rho$ be the grid ratio.  For every positive singleton reward
$a_i$, search the grid $b_i,b_iq,b_iq^2,\ldots$, where
$b_i=\min\{1,U/a_i\}$, far enough to cover multiplicative width $2n^2$,
and query one further point.  For the correct anchor, let $k\geq0$ be the
least index such that $bq^k\leq\delta^*$.  If $k=0$, then $b=\delta^*$;
otherwise minimality gives $q\delta^*<bq^k\leq\delta^*$.  Hence the
additional point $\delta=bq^{k+1}$ satisfies
\[
                         q^2\delta^*<\delta\leq q\delta^*. \tag{11}
\]
This deliberate undershoot preserves at least a $q^2$ fraction of the
principal's retained share while creating the margin
$\alpha-\alpha^*\geq\rho\delta^*$ needed to absorb response error.
Let $\alpha=1-\delta$ and let $\widetilde S_\alpha$ be any response returned
by the approximate oracle.  Approximate optimality at $\alpha$ and exact
optimality of $S^*$ at $\alpha^*$ give
\begin{align*}
\alpha f(\widetilde S_\alpha)-c(\widetilde S_\alpha)
 &\geq \alpha f(S^*)-c(S^*)-\tau,\\
\alpha^*f(S^*)-c(S^*)
 &\geq \alpha^*f(\widetilde S_\alpha)
          -c(\widetilde S_\alpha).
\end{align*}
Adding and using
$\alpha-\alpha^*=\delta^*-\delta\geq\rho\delta^*$ yields
\[
f(\widetilde S_\alpha)
\geq f(S^*)-\frac{\tau}{\rho\delta^*}.
\]
Consequently, since $\delta>q^2\delta^*$ and $\delta\leq q\delta^*$,
\[
\delta f(\widetilde S_\alpha)
 \geq q^2P^*-\frac{q\tau}{\rho}
 \geq (1-\eps)P^*-\frac{3\tau}{\eps}.
\]
The last inequality uses $q^2=(1-\eps/3)^2\geq1-\eps$.
Returning the queried pair with greatest realized principal utility can
only improve this bound.

Finally, covering a multiplicative interval of width $2n^2$ with ratio
$q=1-\eps/3$ takes $O(\log(n+1)/\eps)$ queries per singleton anchor.  The
algorithm therefore uses the stated number of $\tau$-best-response and
reward-value queries, plus the one cost-value query used to compute $L$.
\end{proof}

\clearpage
\section{Equal-revenue family and high-accuracy lower bound}
\label{app:high-accuracy}\label{sec:hard-instance}

The equal-revenue construction underlying the exact lower bound
of~\cite{DFGR26} makes the separation in \cref{cor:separation} especially
transparent.  It also provides an independent stress test for the two
structural quantities used by our algorithm.

Let $N=2^n-1$ be the largest binary index.  Assign each set $S\subseteq[n]$
its binary value $\sum_{i\in S}2^{i-1}$, and let $S_t$ denote the unique set
with index $t\in\{0,\ldots,N\}$; in particular,
$S_0=\varnothing$.  Give item $i$ reward $v_i=2^{i-1}$, so that
$f(S_t)=t$.  The
equal-revenue cost is defined by
\[
 c(S_0)=0,\qquad
 c(S_t)-c(S_{t-1})=\frac{t-1}{t}
 \quad(1\leq t\leq N).
\]
It is monotone and supermodular.  For each $S_t$, define its critical contract
share by
\[
                        \alpha_t=\frac{t-1}{t}.
\]
This contract induces $S_t$ for $1\leq t\leq N$.  As shown below, every such critical contract gives the
principal the same utility, hence the name \emph{equal revenue}.

\begin{proposition}[Anatomy of the equal-revenue instance]
\label{prop:equal-revenue}
On the preceding instance,
\[
       P^*=1,\qquad
       W=H_{2^n-1},
\]
where $H_m=\sum_{h=1}^m1/h$ is the $m$th harmonic number.  Moreover, if
$1\leq t\leq N$ and $j_t$ denotes a largest-reward action in $S_t$, i.e.,
$j_t\in\arg\max_{j\in S_t}v_j$, then
\[
 \frac{1}{2v_{j_t}}
 < 1-\alpha_t
 \leq\frac{1}{v_{j_t}}.                                 \tag{12}
\]
Thus the welfare--profit gap is $\Theta(n)$, while a singleton reward
locates the scale of every critical contract within a factor two.
\end{proposition}

\begin{proof}
By construction,
\[
 (1-\alpha_t)f(S_t)=\frac1t\cdot t=1.
\]
To identify the response curve, let
$u_t(\alpha)=\alpha t-c(S_t)$ be the agent's utility from $S_t$.
The consecutive lines $u_{t-1}$ and $u_t$ intersect at
\[
 c(S_t)-c(S_{t-1})=\alpha_t.
\]
These intersection points strictly increase with $t$, so the upper envelope
visits $S_0,S_1,\ldots,S_N$ in this order.  Thus the $\alpha_t$ are exactly
the critical contracts, and $P^*=1$.  Telescoping the cost recurrence gives
\[
 c(S_t)=\sum_{h=1}^t\frac{h-1}{h}
       =t-H_t.
\]
Hence $f(S_t)-c(S_t)=H_t$, which is maximized at $t=N$.

If $2^{k-1}$ is the highest power of two appearing in the binary
representation of $t$, then $v_{j_t}=2^{k-1}\leq t<2^k=2v_{j_t}$.
Since $1-\alpha_t=1/t$, taking reciprocals yields~(12).
\end{proof}

The lower bound lowers the cost of one hidden set, making one of exponentially
many nearly tied contracts uniquely optimal.  Exact optimization must discover
that set.  At coarser relative accuracy, its small utility advantage can be
ignored: \cref{prop:equal-revenue} shows that every candidate breakpoint has
a scale certified by one singleton reward, which our grid can find without
learning its identity.

The next subsection locates the limit of this explanation.  The hidden cost
perturbation is $\Theta(4^{-n})$, but because it applies to a set of reward
$\Theta(2^n)$, it creates a relative principal-utility advantage
$\Theta(2^{-n})$.  At accuracy comparable to that advantage, approximation
must again identify the hidden set.

\subsection{The high-accuracy query lower bound}

The exact lower bound alone does not explain whether the dependence on
$1/\eps$ in our FPTAS is necessary.  The same hidden-perturbation family
gives a partial answer: once the requested relative error reaches the scale
of the hidden perturbation, approximation again requires exponentially many
queries.

We will use one simulation fact from Appendix~E
of~\cite{DFGR26}.  For a nonnegative price vector $p$, let
$p(S)=\sum_{i\in S}p_i$ be the total price of $S$.  The $\sigma$-approximate
supply ${\cal D}_\sigma(p)$ is the collection of sets whose price-minus-cost
objective lies within $\sigma$ of optimal:
\[
 {\cal D}_\sigma(p)=
 \left\{S\subseteq[n]:
 p(S)-c(S)\geq
 \max_{T\subseteq[n]}\{p(T)-c(T)\}-\sigma\right\}.
\]
For the equal-revenue cost above, if
\[
 0<\sigma<
 \frac12\min_{1\leq t<N}(\alpha_{t+1}-\alpha_t),
\]
then ${\cal D}_\sigma(p)$ has size $O(n^2)$ for every $p$ and can be
enumerated from the known base cost without further oracle calls, though
the enumeration need not be computationally efficient.  Consequently, if a
perturbed cost
differs from $c$ on one set by at most $\sigma$, every exact supply response
to the perturbed cost lies in ${\cal D}_\sigma(p)$; querying the perturbed
cost on these candidates simulates the supply query with $O(n^2)$ value
queries.  This is the sparse-supply property used below.

\begin{proof}[Proof of \cref{thm:approx-lower}]
Fix $n\geq2$, let $N=2^n-1$, and use the equal-revenue instance from
\cref{sec:hard-instance}.  Thus $f(S_t)=t$,
$\alpha_t=(t-1)/t$, and every critical contract gives principal utility
one.  Let ${\cal K}$ be the candidate set of hidden indices and let $\zeta$
be the amount by which the hidden set's cost is reduced:
\[
 {\cal K}=\{\lceil N/2\rceil,\ldots,N\},
 \qquad
 \zeta=\frac{1}{8N^2}.
\]
For each hidden index $k\in{\cal K}$, define $c_k$ by lowering only the cost
of $S_k$ by $\zeta$.
Since
\[
 \alpha_{t+1}-\alpha_t=\frac{1}{t(t+1)}
 \quad\text{and}\quad
 \zeta=\frac{1}{8N^2}
 <\frac{1}{2N(N-1)},
\]
the preceding sparse-supply property applies with $\sigma=\zeta$.

These perturbations preserve nonnegativity, monotonicity, and
supermodularity.  For completeness, write
$c(S_t)=t-H_t$.  Its consecutive increments
$d_t=c(S_t)-c(S_{t-1})=1-1/t$ are increasing.  Every strict
supermodular marginal gap has the following form.  If the indices of
$S\subset T$ differ by $h\geq1$ and the added action has binary value $a$,
then the gap is
\[
 \sum_{r=1}^{a}\bigl(d_{x+h+r}-d_{x+r}\bigr)
 =
 \sum_{r=1}^{a}\frac{h}{(x+r)(x+h+r)}
 \geq\frac1{N^2},
\]
where $x$ is the index of $S$.  The four sets in such a marginal comparison
are distinct, so changing one set value alters this positive gap by at most
$\zeta<1/N^2$.  A monotonicity comparison involving the perturbed set has
gap at least $1-1/k\geq1/2$ before perturbation and therefore remains valid;
all other inequalities are unchanged.

The line for $S_k$ moves upward by $\zeta$.  Its intersection with the
preceding response line shifts from $\alpha_k$ to
$\alpha'_k=\alpha_k-\zeta$, because consecutive rewards differ by one.
Hence its principal utility at entry becomes
\[
                 (1-\alpha'_k)f(S_k)=1+\zeta k
                 \geq1+\frac{1}{16N}.
\]
The small-perturbation inequality above preserves the order of the response
lines.  The lifted line for $S_k$ enters at $\alpha'_k$.  If $k<N$, it
remains active until its intersection with $S_{k+1}$, which shifts to
$\alpha_{k+1}+\zeta$; if $k=N$, it remains active through $\alpha=1$.
Every other response line enters at its original breakpoint, except
$S_{k+1}$ when that line exists, which enters later.  Principal utility
decreases with $\alpha$ while a fixed response is active.  Therefore every
response other than $S_k$ yields utility at most one, whereas $S_k$ yields
its strictly larger maximum at $\alpha'_k$.  This makes $\alpha'_k$ the
unique optimal contract under the stipulated tie-breaking.

Set the target relative error to $\gamma_n=1/(32N)$.  Since
\[
 (1-\gamma_n)\left(1+\frac1{16N}\right)>1,
\]
every successful $(1-\gamma_n)$-approximation must induce $S_k$ and
therefore identifies the hidden index.

It remains to transfer this identification requirement to supply queries.
If the base cost $c$ is known, a value query distinguishes $c_k$ from $c$
only when it queries the hidden set.  With $k$ uniform in ${\cal K}$,
identification therefore needs $\Omega(N)$ value queries by Yao's
principle.  By the sparse-supply property established above, a supply query
to $c_k$ can be simulated by enumerating ${\cal D}_\zeta(p)$ and making
$O(n^2)$ value queries to $c_k$.  Indeed, an optimizer under $c_k$ is within
$\zeta$ of optimal under $c$, because the perturbation can improve the
objective of only one set and by only $\zeta$.  Consequently, an algorithm
making $Q$ supply
queries and polynomially many additional value queries yields
$O(Qn^2)+\operatorname{poly}(n)$ value queries.  The $\Omega(N)$
identification lower bound implies
\[
                    Q\geq\frac{N}{\operatorname{poly}(n)}
                     =2^{\Omega(n)}.
\]
\end{proof}

\clearpage
\small
\section{Tightness constructions}\label{app:tightness}

\begin{proof}[Proof of \cref{prop:tightness}, part~(i)]
Let the ground set be
$A\mathbin{\dot\cup}\{b_1,\ldots,b_m\}$, where $A$ contains $m$
unit-reward actions and action $b_r$ has reward $R_r=2^rm$ for $r\in[m]$.
Rewards are additive.
Set
\[
 R_0=m,\qquad P_0=1+\eta,\qquad
 \alpha_0=1-\frac{P_0}{m},\qquad c(A)=m-P_0.
\]
For $r\geq1$, recursively define
\[
\alpha_r=1-\frac1{R_r},
\qquad
c(\{b_r\})=c(\{b_{r-1}\})
  +\alpha_r(R_r-R_{r-1}),
\]
where $\{b_0\}$ denotes $A$ and $c(\{b_0\})=c(A)$.
Every other nonempty set $S$ receives cost $c(S)=2f(S)$.
These costs are nonnegative.  Dividing every reward and cost by $f([n])$
leaves all responses and ratios unchanged if rewards must lie in $[0,1]$.

The only lines that can appear on the agent's upper envelope are
$\varnothing,A,\{b_1\},\ldots,\{b_m\}$.  Their slopes are strictly
increasing, consecutive lines intersect at
\[
0<\alpha_0<\alpha_1<\cdots<\alpha_m<1.
\]
Thus they appear in this order.  The principal receives $P_0=1+\eta$ from
$A$ at $\alpha_0$, but only $(1-\alpha_r)R_r=1$ at every later
intersection; utility decreases between intersections.  Hence the unique
optimum is $(\alpha^*,S^*)=(\alpha_0,A)$.

Let $W_0=P_0$ and $W_r=R_r-c(\{b_r\})$.  The recurrence gives
\[
W_r-W_{r-1}
=(1-\alpha_r)(R_r-R_{r-1})
 =\frac12.
\]
Therefore $W=P_0+m/2$, while $1-\alpha^*=P_0/m$ and
$\max_{i\in S^*}f(\{i\})=1$.  Substitution proves part~(i).
\end{proof}

\begin{proof}[Proof of \cref{prop:tightness}, part~(ii)]
Let
\[
 Z=\sum_{h=0}^{n-1}r^h,\qquad
 v_i=\frac{r^{i-1}}{Z},\qquad
 P=v_1,\qquad
 a_i=1-\frac{P}{v_i}=1-r^{-(i-1)}.
\]
Define $d_1=0$ and, for $i\geq2$,
\[
 d_i=d_{i-1}+a_i(v_i-v_{i-1}),
\]
and set
\[
 c(S)=\sum_{i\in S}d_i+2\binom{|S|}{2}.
\]
The marginal cost of adding $i$ to $S$ is $d_i+2|S|$, so $c$ is monotone
and supermodular.  Sets of size at least two have cost at least $2$ but
$\alpha f(S)\leq1$, so they are never best responses.  Since
$(d_i-d_{i-1})/(v_i-v_{i-1})=a_i$ and the $a_i$ increase, the upper envelope
visits the singletons in order.  Each gives principal utility
$(1-a_i)v_i=P$, so $P^*=P$.  Finally, telescoping at $\alpha=1$ gives
\begin{align*}
W=v_n-d_n
 &=P+\sum_{i=2}^n(1-a_i)(v_i-v_{i-1})\\
 &=P\left(1+(n-1)(1-1/r)\right)
  =P\left(n-\frac{n-1}{r}\right).
\end{align*}
\end{proof}

\end{document}